\documentclass[final,twoside,11pt]{entics}
\usepackage[utf8]{inputenc}
\usepackage{enticsmacro}
\usepackage{mathtools}
\usepackage{stmaryrd}
\usepackage{tipa}

\DeclareMathOperator{\deq}{{: \hspace{-0.125em} \equiv}}
\DeclareMathOperator{\fresh}{\text{is fresh wrt.}}
\DeclareMathOperator{\grdef}{{: \hspace{-0.125em} : \hspace{-0.125em} =}}
\DeclareMathOperator{\ALT}{|}

\DeclareMathOperator{\FUN}{\lambda}
\DeclareMathOperator{\DOT}{.}
\newcommand*{\Fun}[2]{\FUN #1 \DOT #2}
\newcommand*{\Letk}[3]{\Fun{k}{#2 \, (\Fun{#1}{#3 \, k})}}

\DeclareMathOperator{\LET}{\mathrm{let}}
\DeclareMathOperator{\IN}{\mathrm{in}}
\newcommand*{\Let}[3]{\LET #1 = #2 \IN #3}

\DeclareMathOperator{\IF}{\text{if}}

\DeclareMathOperator{\OTH}{\text{otherwise}}

\DeclareMathOperator{\DOL}{\$}
\newcommand*{\Dol}[1]{\DOL(#1)}

\DeclareMathOperator{\SH}{{\mathcal{S}_0}}
\newcommand*{\Sh}[1]{\SH(#1)}
\newcommand*{\Shift}[2]{\SH #1 \DOT #2}

\DeclareMathOperator{\SHDF}{\mathcal{S}}

\newcommand*{\ld}{{\lambda_\$}}

\newcommand*{\Ld}{{\Lambda_\$}}

\newcommand*{\Lmv}{{\Lambda\mu_v}}
\newcommand*{\lmv}{{\lambda\mu_v}}
\newcommand*{\lcS}{{\lambda_{c\SHDF}}}
\newcommand*{\lcd}{{\lambda_{c\$}}}

\newcommand*{\expa}{\twoheadleftarrow}
\newcommand*{\redu}{\twoheadrightarrow}

\newcommand*{\bv}{{\beta_v}}
\newcommand*{\ev}{{\eta_v}}
\newcommand*{\bd}{{\beta_\$}}
\newcommand*{\ed}{{\eta_\$}}
\newcommand*{\dolv}{{\$_v}}
\newcommand*{\dolsh}{{\DOL\SH}}
\newcommand*{\shdol}{{\SH\DOL}}
\newcommand*{\pure}{\mathbf{pure}}
\newcommand*{\bind}{\mathbf{bind}}

\newcommand*{\cps}[1]{
  {\mathcal{C}\hspace{-0.25em}\left\llbracket #1 \right\rrbracket}}
\newcommand*{\cpsv}[1]{
  {\mathcal{C}_v\hspace{-0.25em}\left\llbracket #1 \right\rrbracket}}
  
  \def\conf{MFPS 2022} 	
\volume{1}			
\def\volu{1}			
\def\papno{14}			
\def\mydoi{{\normalshape  \href{https://doi.org/10.46298/entics.10502}{10.46298/entics.10502}}}
\def\copynames{M. Pyzik} 
\def\runtitle{Call-By-Name is Just Call-By-Value with Delimited Control}

\def\lastname{Pyzik}

\begin{document}

\begin{frontmatter}
  \title{Call-By-Name Is Just Call-By-Value\\with Delimited Control}
  \author{Mateusz Pyzik\thanksref{myemail}\thanksref{orcid}\thanksref{pronunciation}}
  \address{Institute of Computer Science\\University of Wrocław\\Wrocław, Poland}
  \thanks[ALL]{Special thanks to my PhD programme supervisor Dariusz Biernacki and to my colleague and coauthor of previous works Filip Sieczkowski for invaluable support throughout the whole endeavour.}
  \thanks[myemail]{Email: \href{mailto:matp@cs.uni.wroc.pl}{\texttt \normalshape matp@cs.uni.wroc.pl}}
  \thanks[orcid]{ORCID: \href{https://orcid.org/0000-0002-9978-9536}{0000-0002-9978-9536}}
  \thanks[pronunciation]{IPA: \textipa{[mat\href{https://en.wikipedia.org/wiki/Open-mid_front_unrounded_vowel}{E}u\href{https://en.wikipedia.org/wiki/Voiceless_retroflex_fricative}{\:s} p\href{https://en.wikipedia.org/wiki/Close_central_unrounded_vowel}{1}\href{https://en.wikipedia.org/wiki/Voiced_alveolo-palatal_fricative}{\textctz}ik]}}
\begin{abstract}
  Delimited control operator \texttt{shift0} exhibits versatile capabilities: it can express layered monadic effects, or equivalently, algebraic effects. Little did we know it can express lambda calculus too!
  We present \(\Ld\), a call-by-value lambda calculus extended with \texttt{shift0} and control delimiter \texttt{\$} with carefully crafted reduction theory, such that the lambda calculus with beta and eta reductions can be isomorphically embedded into \(\Ld\) via a right inverse of a continuation-passing style translation.
  While call-by-name reductions of lambda calculus can trivially simulate its call-by-value version, we show that addition of \texttt{shift0} and \texttt{\$} is the golden mean of expressive power that suffices to simulate beta and eta reductions while still admitting a simulation back.
  As a corollary, calculi \(\Lmv\), \(\ld\), \(\Ld\) and \(\lambda\) all correspond equationally.
\end{abstract}
\end{frontmatter}
\section{Introduction}\label{intro}

Delimited control~\cite{DBLP:conf/popl/Felleisen88,DBLP:conf/lfp/DanvyF90} is a computational effect capable of expressing any monadic effect in direct style~\cite{DBLP:conf/popl/Filinski94,DBLP:conf/popl/Filinski99}. Moreover, delimited continuations are closely related to algebraic effects~\cite{DBLP:conf/esop/PlotkinP09,DBLP:journals/entcs/Pretnar15}, as it has been shown that delimited control, in the form of \texttt{shift0} and \texttt{\$}, are mutually macro expressible with algebraic effects with deep handlers~\cite{DBLP:journals/jfp/0002KLP19,DBLP:conf/rta/PirogPS19}. It is as strong as an infinite hierarchy of \texttt{shift}s~\cite{DBLP:conf/aplas/MaterzokB12}. Finally, \texttt{shift0} can recursively express \texttt{control0} in direct style~\cite{DBLP:journals/lisp/Shan07}.

It is a basic observation that operational semantics of call-by-value lambda calculus can be simulated by call-by-name lambda calculus. Every \(\bv\) (\(\ev\)) reduction step can be simulated by a \(\beta\) (\(\eta\)) step on the same term. The latter simulation is mediated by identity relation on terms. In this paper we show that there exists a relation (i.e. CPS translation \(* : \Ld \to \lambda\)) between call-by-value lambda terms extended with delimited control operators and vanilla, call-by-name lambda terms that admits a simulation in both ways, hence a bisimulation. We do not prove the bisimulation directly. Instead, we show that a stronger relationship between these calculi exists, namely a reflection.

In our previous, related works we followed the programme of Sabry and Wadler~\cite{DBLP:journals/toplas/SabryW97}: to seek a Galois connection (or even better, a reflection or isomophism) between reduction theories instead of equational correspondences whenever possible, as a Galois connection is a stronger relationship between calculi that already implies an equational correspondence. In~\cite{DBLP:conf/fscd/BiernackiPS20} we developed such relationship between \texttt{shift} and its CPS, while in~\cite{DBLP:conf/ppdp/BiernackiPS21} we have done that for \texttt{shift0}. This paper improves upon the latter of these and upon undirected axiomatisation of~\cite{DBLP:conf/csl/Materzok13}, as we prove a reflection where the image of CPS is closed under \(\beta\eta\) reduction, hence a "directed axiomatisation". However, in the particular case of CPS translation for \texttt{shift0}-style delimited control it turns out that the reduction-closed image is not a proper subcalculus but the entire \(\lambda\), hence every \(\lambda\)-term is, in fact, in continuation-composing style. A translation back (aka. direct style translation) from \(\lambda\) to \(\Ld\) exists which allows to perfectly embed lambda calculus: image of the direct style translation is a subcalculus of \(\Ld\) isomorphic with \(\lambda\)-calculus.

This paper has two goals.
\begin{enumerate}
  \item to bring about the definitive, fine-grained operational semantics for \texttt{shift0} that can provide, a "directed axiomatisation" of delimited control.
  \item to show that call-by-value with delimited control in the form of \texttt{shift0} is the same thing as call-by-name purely functional programming on the level of operational semantics.
\end{enumerate}

The second point cannot be stressed enough. Call-by-name reduction is exactly like call-by-value with \texttt{shift0}-style delimited control. It is not the case with abortive control nor \texttt{shift} of Danvy and Filinski~\cite{DBLP:conf/lfp/DanvyF90}, which can be deduced from previous axiomatisation efforts~\cite{DBLP:journals/lisp/SabryF93,DBLP:conf/icfp/KameyamaH03}, whose CPS translations reach proper subsets of \(\lambda\)-terms. This puts \texttt{shift0} control operator in a unique spot and as far as we know, it has not been acknowledged yet.

\subsection{Outline}

In Section~\ref{sec:definition} a new calculus of control is presented, named \(\Ld\), with its syntax and semantics. In Section~\ref{sec:well-defined} we show mutual macro-expressibility (and equational correspondence) of \(\Ld\) and \(\ld\), an established calculus with \texttt{shift0} and \texttt{\$}, to support our claim that \(\Ld\) is a call-by-value calculus with \texttt{shift0}-style delimited control. In Section~\ref{sec:reflection} we prove theorems that justify the title of this paper: \(\lambda\) can be reflected into \(\Ld\). As a corollary, calculi \(\Lmv\)~\cite{DBLP:conf/icfp/DownenA14}, \(\ld\)~\cite{DBLP:conf/csl/Materzok13}, \(\Ld\) and \(\lambda\) all correspond equationally. We conclude and discuss related and future work in Section~\ref{sec:conclusion}.

\subsection{Conventions}

In this paper, \(\equiv\) means syntactic equality modulo renaming of bound variables, while \(=\) means term equivalence up to rewrite rules in a calculus; calculus is either implied by the usage or specified directly in the subscript. Given relation \(r\), \(r^?\) is \((\equiv) \cup r\), i.e. reflexive closure of \(r\). Rewrites may occur in arbitrary contexts, including variable-binding contexts; these contexts may capture variables of a term that is plugged into. Terms are not required to be closed; it is perflectly fine for them to be open. Some classes of contexts (i.e. bindable and pure contexts) that we define may look a lot like evaluation contexts but that is not their purpose here. We are not concerned with standard reduction nor evaluation strategy, only with general reduction on open terms that may and will often be nondeterministic.

Atomic terms include variables \(x\), parenthesised terms \((M)\), unary
 \texttt{shift0} \(\Sh{M}\) (we call it \emph{thaw}) and unary \texttt{\$} \(\Dol{M}\) (we call it \emph {freeze}). Parentheses in these unary operators are mandatory and there is no space between operator and parenthesis.
Plugging term in a context \(C[M]\) and postfix translations \(M^*,M^\dagger,M^\#,M^\natural\)\footnote{We follow the translation naming convention of Sabry and Wadler~\cite{DBLP:journals/toplas/SabryW97}.} are left-associative and have the highest precedence.
Application \(M \, N\) has the next precedence and is left-associative,
followed by right-associative binary \texttt{\$} (known as \emph{plug}~\cite{DBLP:conf/tlca/KiselyovS07} or \emph{dollar}~\cite{DBLP:conf/aplas/MaterzokB12}\footnote{The word \emph{plug} is somewhat justified, but we feel like it is already taken (to \emph{plug} a term in a context), so we tend to pronouce this operator \emph{dollar}. It's arguably a bad name but there is a precedent of its use.}) \(M \DOL N\).
We disambiguate \(M \, \Dol{N}\) as \(M \, (\Dol{N})\), and \(M \DOL \, (N)\) as \(M \DOL N\) by different spacing.
Syntaxes \(\Let{x}{M}{N}\), \(\Fun{x}{N}\) and \(\Shift{x}{N}\) are right-associative, bind variable \(x\) in term \(N\) and have the lowest precedence.

\subsection{Theory of preorders -- Galois connections and reflections}
\label{sec:galois}

A calculus equipped with a reduction relation can be seen as a preordered set (\emph{proset}, a set with a reflexive and transitive relation).

\begin{definition}
  Assume \((A, \redu_A)\) and \((B, \redu_B)\) are prosets.
  A function \(f : A \to B\) is monotone iff for all \(x_1, x_2 \in A\), \(x_1 \redu_A x_2\) implies \(f(x_1) \redu_B f(x_2)\).
\end{definition}

\begin{definition}
  \label{def:galois}
  Monotone functions \(f : A \to B\) and \(g : B \to A\) form a
  Galois connection if, and only if for all \(a \in A\) and \(b \in B\),
  \(a \redu_A g(b) \Longleftrightarrow f(a) \redu_B b\).
\end{definition}

It may be helpful to think about a Galois connection as a compiler \(*\) from source language \(S\) to target language \(T\) coupled with a decompiler \(\#\). Laws of Galois connection guarantee that compiling and decompiling are harmonised: it does not matter whether simplification is performed on source or target, or on both ends. As an equivalence, a Galois connection allows two modes of reasoning:
\begin{itemize}
  \item \emph{Soundness.} Every multi-step reduction in the source is valid in the target: if \(M \redu_S N^\#\), then \(M^* \redu_T N\).
  \item \emph{Completeness.} Every multi-step reduction in the target is valid in the source: if \(M^* \redu_T N\), then \(M \redu_S N^\#\).
\end{itemize}
A reader familiar with (bi)simulations may notice that if a function \(f\) is monotone, then \(f\) is a similarity relation. When \((f, g)\) is a Galois connection, then \(f\) is a bisimilarity relation.
A reader somewhat familiar with category theory may recognise that preordered set is the same thing as a \emph{thin} category. Monotone functions are functors between those categories. Galois connection is a pair of adjoint functors.
A Galois connection between symmetric preorders (better known as equivalence relations) is simply an equational correspondence.

There is an alternative characterisation of a Galois connection.

\begin{theorem}[Equivalent definition of Galois connection]
  \label{thm:galois-alt-def}
  Monotone functions \(f : A \to B\) and \(g : B \to A\) form a
  Galois connection if, and only if
  \begin{itemize}
  \item \( a \twoheadrightarrow_A g(f(a)) \) and
  \item \( f(g(b)) \twoheadrightarrow_B b \).
  \end{itemize}
\end{theorem}

A reflection is a bit tighter kind of a Galois connection.

\begin{definition}
  A Galois connection \((f : A \to B, g : B \to A)\) is a reflection
  if, and only if for all \(b \in B\), \(f(g(b)) \equiv b\).
\end{definition}

When \((f : A \to B, g : B \to A)\) is a reflection we say that \(B\) can be reflected into \(A\). A subcalculus \((g[B], (\redu_A) \cap (g[B] \times g[B]))\) is then isomorphic to \(B\).

\begin{definition}
  A reflection \( (f : A \to B, g : B \to A) \) is an isomorphism
  if, and only if for all \(a \in A\), \(a \equiv g(f(a))\).
\end{definition}

\subsection{Background: Materzok's \(\ld\)}

\begin{figure}[tp!]
  \label{fig:calc-ld}
  \centering
  \begin{gather*}
    \begin{matrix*}[l]
      \text{term}         & e & \grdef &
        v \ALT e \, e' \ALT \Shift{x}{e} \ALT e \DOL e' \\
      \text{value}        & v & \grdef &
        x \ALT \Fun{x}{M} \\
      \text{pure context} & E & \grdef &
        [\,] \ALT E \, e \ALT v \, E \ALT E \DOL e \\
    \end{matrix*} \\[1mm]
    \begin{matrix*}[l]
      (\Fun{x}{e}) \, v      & \to_\bv & e[v/x] & \\
      v \DOL v'              & \to_\dolv & v \, v' & \\
      v \DOL E[\Shift{x}{e}] & \to_{\$/\SH} &
      e[\Fun{y}{v \DOL E[y]}/x] & \IF y \fresh v, E \\
      (\Fun{x}{e}) \, v & =_\bv    & e[v/x] & \\
      \Fun{x}{v \, x}   & =_\ev    & v & \IF x \fresh v \\
      v \DOL v'         & =_\dolv  & v \, v' & \\
      v \DOL E[e]       & =_{\$_E} & (\Fun{y}{v \DOL E[y]}) \DOL e & \\
      v \DOL \Shift{x}{e} & =_\bd & e[v/x] &\\
      \Shift{x}{x \DOL e} & =_\ed & e & \IF x \fresh e \\
    \end{matrix*}
  \end{gather*}
  \caption{Syntax, reductions and axioms of \(\ld\)~\cite{DBLP:conf/csl/Materzok13}.}
\end{figure}

\begin{figure}[tp!]
  \centering
  \begin{gather*}
    \cps{\cdot} : \ld \to \lambda \\
    \begin{matrix*}[l]
      \cps{v} & \deq \Fun{k}{k \, \cpsv{v}} \\
      \cps{e \, e'} & \deq \Fun{k}{\cps{e} \,
      (\Fun{x}{\cps{e'} \, (\Fun{y}{x \, y \, k})})} \\
      \cps{\Shift{x}{e}} & \deq \Fun{x}{\cps{e}} \\      
      \cps{e \DOL e'} & \deq \Fun{k}{\cps{e} \,
      (\Fun{x}{\cps{e'} \, x \, k})} \\
      \cpsv{x} & \deq x \\
      \cpsv{\Fun{x}{e}} & \deq \Fun{x}{\cps{e}} \\
    \end{matrix*}\\
    \begin{matrix*}[l]
      \Sh{e}  & \deq & (\Fun{x}{\Shift{k}{x \, k}}) \, e \\
      \Dol{e} & \deq & \Fun{x}{x \DOL e} \\
    \end{matrix*}
  \end{gather*}
  \caption{Semantics of \(\ld\), described by CPS translation~\cite{DBLP:conf/csl/Materzok13} and macro-definitions of unary variants of control operators.}
  \label{fig:cps-ld}
\end{figure}

In this introductory subsection, we present syntax, reductions, axioms (Figure~\ref{fig:calc-ld}) and CPS translation (Figure~\ref{fig:cps-ld}) of \(\ld\), as defined by Materzok~\cite{DBLP:conf/csl/Materzok13}. It features a variable-binding control operator \texttt{shift0} and binary operator \texttt{\$} that delimits continuation captured by \texttt{shift0} in the right hand side. CPS translation to lambda calculus provides denotational semantics. The axioms were proven by Materzok to be sound and complete with respect to CPS translation. \(\ld\) also has small-step operational semantics, which are sound (but not complete) with respect to axioms and CPS. These semantics allow \texttt{shift0} to capture an arbitrarily long context in a single reduction step.

\begin{example}
  Here's an evaluation example.
  Let \(I \equiv \Fun{x}{x}\).
  \begin{gather*}
    I \DOL I \, (\Shift{f}{f \, (f z)})
    \to_{\$/\SH} (f \, (f \, z))[\Fun{x}{I \DOL I \, x}/f]
    \equiv (\Fun{x}{I \DOL I \, x}) \, ((\Fun{x}{I \DOL I \, x}) \, z) \\
    \to_\bv (\Fun{x}{I \DOL I \, x}) \, (I \DOL I \, z)
    \to_\bv (\Fun{x}{I \DOL I \, x}) \, (I \DOL z)
    \to_\dolv (\Fun{x}{I \DOL I \, x}) \, (I \, z) \\
    \to_\bv (\Fun{x}{I \DOL I \, x}) \, z
    \to_\bv I \DOL I \, z
    \to_\bv I \DOL z
    \to_\dolv I \, z
    \to_\bv z
  \end{gather*}
\end{example}

\begin{example}
  Terms in CPS get lengthy even for very small examples:
  \(\cps{I} \equiv \Fun{k_1}{k_1 \, (\Fun{x}{\Fun{k_2}{k_2 \, x}})}\).
\end{example}

\section{Syntax and semantics of \(\Ld\)}
\label{sec:definition}

Reductions in Materzok's \(\ld\) were clearly lacking. What we wanted was a set of reductions that would also stand as a complete axiomatisation of \texttt{shift0}.
Our goal was to improve upon a sound but incomplete reduction theory of \(\lcd\)~\cite{DBLP:conf/ppdp/BiernackiPS21}. To reach completeness, we took the range of that previous CPS translation and closed this set under \(\beta\eta\)-reduction (in a similar vein to Sabry and Felleisen~\cite{DBLP:journals/lisp/SabryF93}). To our surprise, all quirks of the grammar disappeared and the entire \(\lambda\) showed up. To double check, we also closed CPS of \(\ld\) under reduction: entire \(\lambda\) again.
It was now a matter of time to find an improved calculus, CPS and DS (direct style) translation. Our choice to use unary operators may be considered arbitrary by some but this choice provided us with a beautiful duality of values and terms mediated by freeze and thaw and an actual, working reflection. Old semantics had to be modified anyway and syntax could use some refreshment too. Our method was to have only one binding construct and use only local reductions to keep things as simple as possible.

We define syntax and small-step operational semantics of \(\Ld\) in Figure~\ref{fig:calc-Ld}.
Figure also includes some syntactic sugar to make some idioms (like let-expressions) easier on the eye.

Indeed, let-expressions with both \(\beta\) and \(\eta\) rule are macro-definable.
It is easy to check that \(\Let{x}{V}{M} \redu_\Ld M[V/x]\), \(\Let{x}{M}{x} \redu_\Ld M\) and also \((\Fun{x}{N}) \, M \redu_\Ld \Let{x}{M}{N}\).
Associativity of let-expressions holds, but it is undirected: if \(y \fresh N\), then \[\Let{x}{\Let{y}{L}{M}}{N} =_\Ld \Let{y}{L}{\Let{x}{M}{N}}.\]

Initially we aimed for a calculus with a separate let construct whose CPS is diffrent than that of a \(\beta\)-redex,
like in \(\lambda_c\)~\cite{DBLP:journals/toplas/SabryW97}, \(\lcS\)~\cite{DBLP:conf/fscd/BiernackiPS20} and \(\lcd\)~\cite{DBLP:conf/ppdp/BiernackiPS21}.
However, associativity kept derailing our theorems. We were stuck, so we got rid of the let construct and replaced its uses with a pattern \(\Shift{k}{(\Fun{x}{k \DOL N}) \DOL M}\) that behaves just like let-expressions. Due to its ubiquity in the proofs, we needed a shorthand notation; a let-expression syntax was an obvious choice.

\begin{figure}[tp!]
  \centering
  \begin{gather*}
    \begin{matrix*}[l]
      \text{term}         & L,M,N & \grdef & V \ALT P \\
      \text{value}        & V,W   & \grdef &
        x \ALT \Fun{x}{M} \ALT \Dol{M} \\
      \text{nonvalue}     & P,Q   & \grdef &
        M \, N \ALT \Sh{M} \\
      \text{bindable  context} & J & \grdef &
        [\,] \, M \ALT V \, [\,] \ALT \, \Sh{[\,]} \\
      \text{pure      context} & K & \grdef &
        [\,] \ALT J[K] \\
      \text{reduction context} & C & \grdef &
        [\,] \ALT \Fun{x}{C} \ALT C \, M \ALT
        M \, C \ALT \Sh{C} \ALT \Dol{C} \\
    \end{matrix*} \\
    \begin{matrix*}[l]
      \Shift{x}{M}  & \deq & \Sh{\Fun{x}{M}} \\
      M \DOL N      & \deq & \Dol{N} \, M \\
      \Let{x}{M}{N} & \deq & \Shift{k}{(\Fun{x}{k \DOL N}) \DOL M} \\
    \end{matrix*} \\
    \mapsto \deq \bv \cup \ev \cup
      \dolv \cup \dolsh \cup \shdol \cup \pure \cup \bind \\
    \begin{matrix*}[l]
      (\Fun{x}{M}) \, V & \bv & M[V/x] &\\
      \Fun{x}{V \, x}   & \ev & V & \IF x \fresh V \\
      \Dol{V}           & \dolv & \Fun{x}{x \, V} & \IF x \fresh V \\
      \Dol{\Sh{V}}      & \dolsh & V &\\
      \Sh{\Dol{M}}      & \shdol & M &\\
      \Shift{x}{x \, V} & \pure & V & \IF x \fresh V \\
      J[P]              & \bind & \Let{x}{P}{J[x]} & \IF x \fresh J \\
      C[M] & \to & C[N] & \IF M \mapsto N \\
      M & \redu & N & \IF M \to^n N, n \geq 0 \\
      M & = & N & \IF M (\to\cup\leftarrow)^n N, n \geq 0 \\
    \end{matrix*}
  \end{gather*}
  \caption{Direct style calculus \(\Ld\).}
  \label{fig:calc-Ld}
\end{figure}

One can think of unary \texttt{\$} as reifying or freezing a computation as a value. Such a value can be then reflected or thawed by unary \texttt{shift0} (if given a nonvalue, it waits until its argument becomes a value). These operations are invertible using \(\dolsh\) and \(\shdol\) reductions. Reductions \(\dolv\) and \(\pure\) tell us that a value \(V\) freezes to \(\Fun{x}{x \, V}\) and can be recovered from it by thawing.
The job of \(\bind\) rule is to gives names to subcomputations and make sure contexts are managed properly. Notice that due to syntactic restriction that requires \(\bind\)-redex to be of the form \(J[P]\), there is at most one legal way to \(\bind\)-contract:
\begin{itemize}
  \item \(\Sh{P} \, \bind \Let{x}{P}{\Sh{x}}\) is a legal $\bind$-contraction,
  \item \(\Sh{V} \, \bind \Let{x}{V}{\Sh{x}}\), \(V \, W \, \bind \, \Let{x}{W}{V \, x}\) and \(V \, W \, \bind \, \Let{x}{V}{x \, W}\) are illegal,
  \item \(V \, P \, \bind \, \Let{x}{P}{V \, x}\) is legal while \(V \, P \, \bind \, \Let{x}{V}{x \, P}\) is not,
  \item \(P \, M \, \bind \, \Let{x}{P}{x \, M}\) is legal while \(P \, Q \, \bind \, \Let{x}{Q}{P \, x}\) is not,
\end{itemize}

To enable further developments, Figure~\ref{fig:asterisk} introduces continuation-passing style translation \(* : \Ld \to \lambda\), which provides alternative, denotational semantics of \(\Ld\).
However, one does not need to worry about a mismatch of semantics -- theorems of Section~\ref{sec:reflection} are more than enough to ensure that \(*\) is a sound and complete translation (\(M =_\Ld N\) iff \(M^* =_\lambda N^*\)).
\(\lambda\)-calculus that we target has both \(\beta\) and \(\eta\) reductions (Figure~\ref{fig:calc-lambda}).

\begin{figure}[tp!]
  \label{fig:asterisk}
  \[ * : \Ld \to \lambda \]
  \begin{minipage}{.5\textwidth}
    \begin{equation*}
      \begin{matrix*}[l]
        V^* & \deq \Fun{k}{k \, V^\dagger} \\
        J[P]^* & \deq \Letk{x}{P^*}{J[x]^*} \\
        (V \, W)^* & \deq V^\dagger \, W^\dagger \\
        \Sh{V}^* & \deq V^\dagger \\
      \end{matrix*}
    \end{equation*}
  \end{minipage}
  \begin{minipage}{.4\textwidth}
    \begin{equation*}
      \begin{matrix*}[l]
        x^\dagger & \deq x \\
        (\Fun{x}{M})^\dagger & \deq \Fun{x}{M^*} \\
        \Dol{M}^\dagger & \deq M^* \\
        & \\
      \end{matrix*}
    \end{equation*}
  \end{minipage}
  \caption{
    Translation from \(\Ld\) to continuation-passing style.}
\end{figure}

To lift the spirit even higher, confluence of \(\Ld\) is proven below using definitions (translation \(\# : \lambda \to \Ld\)) and theorems of Section~\ref{sec:reflection} but the proof is straightforward enough that it may be understandable right here.
It gives us a taste of usefullness of reflections and Galois connections in general: one can transfer confluence from one calculus to another.
Confluence is provided here as a sanity check of a newly introduced calculus; to avoid circular reasoning, we do not use it in proofs.

\begin{theorem}[Confluence]
  Reduction in \(\Ld\) is confluent.
\end{theorem}

\begin{proof}
  For concise presentation, we chain relations and flip arguments:
  \(M \expa N\) iff \(N \redu M\).
  \begin{enumerate}
    \item Assume \(M_1 \expa M \redu M_2\).
    \item Apply monotonicity of translation \(*\): \(M_1^* \expa M^* \redu M_2^*\).
    \item Apply confluence in \(\lambda\) (Church-Rosser theorem): there is \(N \in \lambda\) such that \(M_1^* \redu N \expa M_2^*\).
    \item Apply monotonicity of translation \(\#\): \(M_1^{*\#} \redu N^\# \expa M_2^{*\#}\).
    \item Apply left post-inverse theorem twice:
    \(M_1 \redu M_1^{*\#} \redu N^\# \expa M_2^{*\#} \expa M_2\).
    \item Apply transitivity of \(\redu_\Ld\) twice:
    \(M_1 \redu N^\# \expa M_2\).
  \end{enumerate}
\end{proof}

The following lemma is quite handy in equational reasoning about \(\Ld\) terms.

\begin{lemma}[Generalised \(=_\bind\)]
  If \(x \fresh J\), then \(J[M] =_\Ld \Let{x}{M}{J[x]}\).
\end{lemma}

\begin{proof}
  By a chain of rewrites.
  \begin{align*}
    & & J[M] \\
    & \leftarrow_\shdol & J[\Sh{\Dol{M}}] \\
    & \rightarrow_\bind & \Let{x}{\Sh{\Dol{M}}}{J[x]} \\
    & \rightarrow_\shdol & \Let{x}{M}{J[x]}
  \end{align*}
\end{proof}

\section{Equational correspondence with \(\ld\)}
\label{sec:well-defined}

In this section, we defend the thesis that \(\Ld\) is a calculus of delimited control and its control primitives behave like control operator \texttt{shift0} and delimiter \texttt{\$}. To meet these ends, we show that there exists an equational correspondence with \(\ld\) of Materzok~\cite{DBLP:conf/csl/Materzok13} expressed with macro-translations. As we mentioned already, \(\ld\) has small-step operational semantics, which are sound (but not complete) with respect to axioms and CPS. These semantics allow \texttt{shift0} to capture an arbitrarily long context in a single reduction step. Equational correspondence that we prove in this section shows that this traditional approach can also be used in equational reasoning about \(\Ld\) terms.

One might wonder, why we only prove equational correspondence and not a Galois connection? We clearly cannot connect directed \(\redu_\Ld\) with \(=_\ld\) because such a connection would transfer symmetry property from \(=_\ld\) to \(\redu_\Ld\), a contradiction: \(I \, I \to I \not{\to} I \, I\). We cannot connect \(\redu_\Ld\) with \(\redu_\ld\) using macros: every macro-translation would necessarily transfer \(\Fun{y}{x\,y} \redu x\) verbatim which is provably true in \(\Ld\) and false in \(\ld\). Our equational correspondence crucially uses Materzok's $\$_E$ axiom in both directions, so it is unclear how these axioms could be turned into directed reduction rules that would admit a Galois connection. We did not investigate translations that are not based on macros: such translation would weaken our claim that \(\Ld\) is a calculus with \texttt{shift0}-style delimited control.

Basic reductions of \(\Ld\) can be seen as set of axioms.
It is later shown in Section~\ref{sec:reflection} as corollary that for such axioms CPS translation \(*\) is sound and \emph{complete} (\(M =_\Ld N\) iff \(M^* =_\lambda N^*\)), which is an improvement upon \(\ld\), whose reductions induce an equivalence relation on \(\ld\) that is a proper subset of \(=_\ld\), hence incomplete. These stronger properties make \(\Ld\) a contender for the title of the definitive calculus of \texttt{shift0}-style delimited control.

To mediate the correspondence, we need translations to (\(\iota\)) and from (\(\pi\)) \(\Ld\).
Translations (with unfolded macros) are presented in Figure~\ref{fig:coercion}.
With folded macro-definitions they would look like an identity function.
To translate from \(\ld\) to \(\Ld\), we use macros from Figure~\ref{fig:calc-Ld}.
To translate from \(\Ld\) to \(\ld\), we use macros from Figure~\ref{fig:cps-ld}.

\begin{figure}[tp!]
  \begin{minipage}{.4\textwidth}
    \begin{gather*}
      \iota : \ld \to \Ld \\
      \begin{matrix*}[l]
        \iota(x) & \deq x \\
        \iota(\Fun{x}{e}) & \deq \Fun{x}{\iota(e)} \\
        \iota(e \, e') & \deq \iota(e) \, \iota(e') \\
        \iota(\Shift{x}{e}) & \deq \Sh{\Fun{x}{\iota(e)}} \\
        \iota(e \DOL e') & \deq \Dol{\iota(e')} \, \iota(e) \\
      \end{matrix*}
    \end{gather*}
  \end{minipage}
  \begin{minipage}{.5\textwidth}
    \begin{gather*}
      \pi : \Ld \to \ld \\
      \begin{matrix*}[l]
        \pi(x) & \deq x \\
        \pi(\Fun{x}{M}) & \deq \Fun{x}{\pi(M)} \\
        \pi(\Dol{M}) & \deq \Fun{x}{x \DOL \pi(M)} \\
        \pi(M \, N) & \deq \pi(M) \, \pi(N) \\
        \pi(\Sh{M}) & \deq (\Fun{x}{\Shift{k}{x\,k}}) \, \pi(M) \\
      \end{matrix*}
    \end{gather*}
  \end{minipage}
  \caption{Embedding \(\iota\) and its inverse \(\pi\), with macros unfolded.}
  \label{fig:coercion}
\end{figure}

\begin{lemma}[Embedding \(\iota\) is invertible]
  The following equalities hold:
  \begin{itemize}
    \item \emph{Left inverse property}. For all \(e \in \ld\), \(\pi(\iota(e)) = e\).
    \item \emph{Right inverse property}. For all \(M \in \Ld\), \(\iota(\pi(M)) = M\).
  \end{itemize}
\end{lemma}

\begin{proof}
  Both propositions are separately proven by structural induction.
  \begin{itemize}
    \item Case \(e \equiv x\). \(\pi(\iota(x)) \equiv \pi(x) \equiv x\).
    \item Case \(e \equiv \Fun{x}{e_1}\). \(\pi(\iota(\Fun{x}{e_1})) \equiv \pi(\Fun{x}{\iota(e_1)}) \equiv \Fun{x}{\pi(\iota(e_1))} =_\text{IH} \Fun{x}{e_1}\).
    \item Case \(e \equiv e_1 \, e_2\). \(\pi(\iota(e_1 \, e_2)) \equiv \pi(\iota(e_1) \, \iota(e_2)) \equiv \pi(\iota(e_1)) \, \pi(\iota(e_2)) =_\text{IH} e_1 \, e_2 \).
    \item Case \(e \equiv \Shift{x}{e_1}\).
    \begin{align*}
      & & \pi(\iota(\Shift{k}{e_1})) \\
      & \equiv & \pi(\Shift{k}{\iota(e_1)}) \\
      & \equiv & (\Fun{x}{\Shift{k}{x\,k}}) \, \pi(\Fun{k}{\iota(e_1)}) \\
      & \equiv & (\Fun{x}{\Shift{k}{x\,k}}) \, \Fun{k}{\pi(\iota(e_1))} \\
      & =_\text{IH} & (\Fun{x}{\Shift{k}{x\,k}}) \, \Fun{k}{e_1} \\
      & =_\bv & \Shift{k}{(\Fun{k}{e_1})\,k} \\
      & =_\bv & \Shift{k}{e_1}
    \end{align*}
    \item Case \(e \equiv e_1 \DOL e_2\).
    \begin{align*}
      & & \pi(\iota(e_1 \DOL e_2)) \\
      & \equiv & \pi(\iota(e_1) \DOL \iota(e_2)) \\
      & \equiv & \pi(\Dol{\iota(e_2)}) \, \pi(\iota(e_1)) \\
      & \equiv & (\Fun{x}{x \DOL \pi(\iota(e_2))}) \, \pi(\iota(e_1)) \\
      & =_\text{IH} & (\Fun{x}{x \DOL e_2}) \, e_1 \\
      & =_\ed & \Shift{k}{k \DOL \, (\Fun{x}{x \DOL e_2}) \, e_1} \\
      & =_{\$_E} & \Shift{k}{(\Fun{x}{k \DOL \, (\Fun{x}{x \DOL e_2}) \, x}) \DOL e_1} \\
      & =_\bv & \Shift{k}{(\Fun{x}{k \DOL x \DOL e_2}) \DOL e_1} \\
      & =_{\$_E} & \Shift{k}{k \DOL e_1 \DOL e_2} \\
      & =_\ed & e_1 \DOL e_2 \\
    \end{align*}
  \item Case \(M \equiv x\). \(\iota(\pi(x)) \equiv \iota(x) \equiv x\).
  \item Case \(M \equiv \Fun{x}{M_1}\). \(\iota(\pi(\Fun{x}{M_1})) \equiv \iota(\Fun{x}{\pi(M_1)}) \equiv \Fun{x}{\iota(\pi(M_1))} =_\text{IH} \Fun{x}{M_1}\).
  \item Case \(M \equiv \Dol{M_1}\).
  \(
    \iota(\pi(\Dol{M_1}))
    \equiv \iota(\Fun{x}{x \DOL \pi(M_1)})
    \equiv \Fun{x}{\iota(x \DOL \pi(M_1))}
    \equiv \Fun{x}{\iota(x) \DOL \iota(\pi(M_1))}
    \equiv \Fun{x}{x \DOL \iota(\pi(M_1))}
    =_\text{IH} \Fun{x}{x \DOL M_1}
    =_\ev \Dol{M_1}
  \).
  \item Case \(M \equiv M_1 M_2\). \(\iota(\pi(M_1 \, M_2)) \equiv \iota(\pi(M_1) \, \pi(M_2)) \equiv \iota(\pi(M_1)) \, \iota(\pi(M_2)) =_\text{IH} M_1 \, M_2 \).
  \item Case \(M \equiv \Sh{M_1}\).
  \begin{align*}
    & & \iota(\pi(\Sh{M_1})) \\
    & \equiv & \iota((\Fun{x}{\Shift{k}{x\,k}}) \, \pi(M_1)) \\
    & \equiv & \iota((\Fun{x}{\Shift{k}{x\,k}})) \, \iota(\pi(M_1)) \\
    & \equiv & (\Fun{x}{\iota(\Shift{k}{x\,k})}) \, \iota(\pi(M_1)) \\
    & \equiv & (\Fun{x}{\Shift{k}{\iota(x\,k)}}) \, \iota(\pi(M_1)) \\
    & \equiv & (\Fun{x}{\Shift{k}{\iota(x)\,\iota(k)}}) \, \iota(\pi(M_1)) \\
    & \equiv & (\Fun{x}{\Shift{k}{x\,k}}) \, \iota(\pi(M_1)) \\
\end{align*}
\begin{align*}
    & =_\text{IH} & (\Fun{x}{\Shift{k}{x\,k}}) \, M_1 \\
    & =_\ev & (\Fun{x}{\Sh{x}}) \, M_1 \\
    & =_\bind & \Let{x}{M_1}{\, (\Fun{x}{\Sh{x}}) \, x} \\
    & =_\bv & \Let{x}{M_1}{\Sh{x}} \\
    & =_\bind & \Sh{M_1}
  \end{align*}
  \end{itemize}
\end{proof}

\begin{lemma}
  \label{lemma:dol-j}
  Equality \(V \DOL J[M] = (\Fun{x}{V \DOL J[x]}) \DOL M\) holds in \(\Ld\).
\end{lemma}

\begin{proof}
  By a chain of rewrites.
  \begin{align*}
    & & V \DOL J[M] \\
    & =_\bind & V \DOL \Shift{k}{(\Fun{y}{k \DOL J[y]}) \DOL M} \\
    & =_\dolsh & (\Fun{k}{(\Fun{y}{k \DOL J[y]}) \DOL M}) \, V \\
    & =_\bv & (\Fun{y}{V \DOL J[y]}) \DOL M
  \end{align*}
\end{proof}

\begin{lemma}
  \label{lemma:dol-e}
  Equality \(V \DOL K[M] =\Ld (\Fun{x}{V \DOL K[x]}) \DOL M\) holds if \(x \fresh V \text{and} \, K\).
\end{lemma}

\begin{proof}
  By induction on \(K\).
  In the context of this proof,
  we use the rules of \(\Ld\) (Fig.~\ref{fig:calc-Ld}).
  \begin{itemize}
    \item Base case.
    \(
      V \DOL M
      =_\ev (\Fun{x}{V \, x}) \DOL M
      =_\bv (\Fun{x}{(\Fun{y}{y \, x}) \, V}) \DOL M
      =_\dolv (\Fun{x}{V \DOL x}) \DOL M
    \).
    \item Inductive case. Use of an inductive hypothesis is marked by "IH".
    \begin{align*}
      & & V \DOL J[K[M]] \\
      & =_\text{Lemma~\ref{lemma:dol-j}} & (\Fun{x}{V \DOL J[x]}) \DOL K[M] \\
      & =_\text{IH} & (\Fun{y}{(\Fun{x}{V \DOL J[x]}) \DOL K[y]}) \DOL M \\
      & =_\text{Lemma~\ref{lemma:dol-j}} & (\Fun{y}{V \DOL J[K[y]]}) \DOL M \\
    \end{align*}
  \end{itemize}
\end{proof}

\begin{theorem}[Soundness of \(\iota\)]
  The following propositions hold:
  \begin{itemize}
    \item For all \(e, e' \in \ld\), if \(e =_\ld e'\), then \(\iota(e) =_\Ld \iota(e')\).
    \item For all \(M, M' \in \Ld\), if \(\pi(M) =_\ld \pi(M')\),
    then \(M =_\Ld M'\).
  \end{itemize}
\end{theorem}

\begin{proof}
  To prove the first proposition it suffices to show that all axioms of \(\ld\) are also valid equalities in \(\Ld\).
  \begin{itemize}
    \item Axiom \(\bv\) follows immediately by rule \(\bv\).
    \item Axiom \(\ev\) follows immediately by rule \(\ev\).
    \item Axiom \(\dolv\).
    \(
      V \DOL W
      \rightarrow_\dolv (\Fun{k}{k \, W}) \, V
      \rightarrow_\bv V \, W
      \).
    \item Axiom \(\$_E\) follows by Lemma~\ref{lemma:dol-e}.
    \item Axiom \(\beta_\$\).
    \(
      V \DOL \Shift{x}{M}
      \rightarrow_\dolsh (\Fun{x}{M}) \, V
      \rightarrow_\bv M[V/x]
    \).
    \item Axiom \(\eta_\$\).
    \(
      \Shift{x}{x \DOL M}
      \rightarrow_\ev \Sh{\Dol{M}}
      \rightarrow_\shdol M
    \).
  \end{itemize}
  To prove the second proposition, apply the first proposition on assumption \(\pi(M) = \pi(M')\) to get \(\iota(\pi(M)) = \iota(\pi(M'))\) and rewrite by the right inverse property of \(\iota\) to infer \(M = M'\).
\end{proof}

Having the soundness of \(\iota\) secured, we move on to completeness. In order to prove it, we show that CPS translation of Materzok is equivalent to ours. Syntax and small-step operational semantics of \(\lambda\)-calculus are provided for reference in Figure~\ref{fig:calc-lambda}.

\begin{figure}[tp!]
  \centering
  \begin{gather*}
    \begin{matrix*}[l]
      \text{term} & M,N & \grdef &
        x \ALT \Fun{x}{M} \ALT M \, N \\
      \text{reduction context} & C & \grdef &
        [\,] \ALT \Fun{x}{C} \ALT C \, M \ALT M \, C \\
    \end{matrix*} \\
    \begin{matrix*}[l]
      (\Fun{x}{M}) \, N & \beta & M[N/x] & \\
      \Fun{x}{M \, x}   & \eta  & M & \IF x \fresh M \\
      M & \mapsto & N & \IF M \beta N \text{ or } M \eta N \\
      C[M] & \to & C[N] & \IF M \mapsto N \\
      M & \redu & N & \IF M \to^n N, n \geq 0 \\
      M & = & N & \IF M (\to\cup\leftarrow)^n N, n \geq 0 \\
    \end{matrix*}
  \end{gather*}
  \caption{Lambda calculus \(\lambda\), the image of CPS translation.}
  \label{fig:calc-lambda}
\end{figure}

\begin{lemma}
  \label{lemma:app-asterisk}
  Equality \((M \, N)^* = \Fun{k}{M^* \, (\Fun{x}{N^* \, (\Fun{y}{x \, y \, k})})}\) holds in \(\lambda\).
\end{lemma}

\begin{proof}
  In the appendix.
\end{proof}

\begin{lemma}[\(\beta\eta\)-equivalence of CPS transforms]
  For all \(e \in \ld\), \(\iota(e)^* =_\lambda \cps{e}\).
\end{lemma}

\begin{proof}
  By structural induction.
  \begin{itemize}
    \item Case \(e \equiv x\).
    \(
      \iota(x)^*
      \equiv x^*
      \equiv \Fun{k}{k \, x^\dagger}
      \equiv \Fun{k}{k \, x}
      \equiv \cps{x}
    \).
    \item Case \(e \equiv \Fun{x}{e_1}\).
    \begin{align*}
      \iota(\Fun{x}{e_1})^*
      & \equiv (\Fun{x}{\iota(e_1)})^*
      \equiv \Fun{k}{k \, (\Fun{x}{\iota(e_1)})^\dagger}
      \equiv \Fun{k}{k \, \Fun{x}{\iota(e_1)^*}} \\
      & =_\text{IH} \Fun{k}{k \, \Fun{x}{\cps{e_1}}}
      \equiv \Fun{k}{k \, \cpsv{\Fun{x}{e_1}}}
      \equiv \cps{\Fun{x}{e_1}}
    \end{align*}
    \item Case \(e \equiv e_1 \, e_2\).
    \begin{align*}
      \iota(e_1 \, e_2)^*
      & \equiv (\iota(e_1) \, \iota(e_2))^*
      =_\text{Lemma~\ref{lemma:app-asterisk}} \Fun{k}{\iota(e_1)^* \, (\Fun{x}{\iota(e_2)^* \, (\Fun{y}{x \, y \, k})})} \\
      & =_\text{IH} \Fun{k}{\cps{e_1} \, (\Fun{x}{\cps{e_2} \, (\Fun{y}{x \, y \, k})})}
      \equiv \cps{e_1 \, e_2}
    \end{align*}
    \item Case \(e \equiv \Shift{x}{e_1}\).
    \begin{align*}
      & & \iota(\Shift{x}{e_1})^* \\
      & \equiv & \Fun{x}{\iota(e_1)^*} \\
      & =_\text{IH} & \Fun{x}{\cps{e_1}} \\
      & \equiv & \cps{\Shift{x}{e_1}} \\
    \end{align*}
    \item Case \(e \equiv e_1 \DOL e_2\).
    \begin{align*}
      & & \iota(e_1 \DOL e_2)^* \\
      & \equiv & (\iota(e_1) \DOL \iota(e_2))^* \\
      & \equiv & (\Dol{\iota(e_2)} \, \iota(e_1))^* \\
      & =_\text{Lemma~\ref{lemma:app-asterisk}} & \Fun{k}{\Dol{\iota(e_2)}^* \, (\Fun{x}{\iota(e_1)^* \, (\Fun{y}{x \, y \, k})})} \\
      & =_\beta & \Fun{k}{(\Fun{x}{\iota(e_1)^* \, (\Fun{y}{x \, y \, k})}) \, \Dol{\iota(e_2)}^\dagger} \\
      \end{align*}
      \begin{align*}
      & =_\beta & \Letk{y}{\iota(e_1)^*}{\Dol{\iota(e_2)}^\dagger \, y} \\
      & \equiv & \Letk{y}{\iota(e_1)^*}{\iota(e_2)^* \, y} \\
      & =_\text{IH} & \Letk{y}{\cps{e_1}}{\cps{e_2} \, y} \\
      & \equiv & \cps{e_1 \DOL e_2} \\
    \end{align*}
  \end{itemize}
\end{proof}

We strike the final nail in the coffin of doubt with the completeness theorem.

\begin{theorem}[Completeness of \(\iota\)]
  The following propositions hold:
  \begin{itemize}
    \item For all \(e, e' \in \ld\), if \(\iota(e) =_\Ld
      \iota(e')\), then \(e =_\ld e'\).
    \item For all \(M, M' \in \Ld\), if \(M =_\Ld M'\),
      then \(\pi(M) =_\ld \pi(M')\).
  \end{itemize}
\end{theorem}

\begin{proof}
  To prove the first proposition, apply monotonicity of \(*\) (Theorem~\ref{thm:mono-asterisk}) on assumption \(\iota(e) = \iota(e')\) to get \(\iota(e)^* = \iota(e')^*\). Rewrite using equivalence of
  CPS translations to get \(\cps{e} = \cps{e'}\) and finally apply completeness of \(\cps{\cdot}\) to conclude that \(e = e'\).
  To prove the second proposition, rewrite assumption \(M = M'\) with left inverse property of \(\iota\) to get \(\iota(\pi(M)) = \iota(\pi(M))\). Apply the first proposition to finally arrive at \(\pi(M) = \pi(M')\).
\end{proof}

It follows that calculi \(\Ld\) and \(\ld\) correspond equationally via macro-definitions, hence our \(\Ld\) is truly a calculus of \texttt{shift0}-style delimited control.

\section{Reflection of lambda calculus into \(\Ld\)}
\label{sec:reflection}

We move on to the main result of this paper: the relationship between \(\Ld\) and \(\lambda\). It will be shown that \(\lambda\) reflects into \(\Ld\).
Syntax and small-step operational semantics of \(\lambda\)-calculus are provided for reference in Figure~\ref{fig:calc-lambda}.
Following in the footsteps of Sabry and Wadler~\cite{DBLP:journals/toplas/SabryW97} we intended to define a backwards, \emph{direct style} translation \(\#\) from the range of \(*\), just like in~\cite{DBLP:conf/fscd/BiernackiPS20} and~\cite{DBLP:conf/ppdp/BiernackiPS21}.
We took the range of the old CPS and closed it under reduction, following the recipe by Sabry and Felleisen~\cite{DBLP:journals/lisp/SabryF93}.
What we've got was the entire set of lambda terms. A question arised: was there a better CPS that would explicitly hit every possible lambda term? The answer is positive.
In other words, for \emph{every} term \(M \in \lambda\), there exists a term \(M^\# \in \Ld\), such that \(M^{\#*} \equiv M\). Every lambda term is in Continuation-Passing Style, it seems!

Translation \(\#\) has one oddity: special treatment of \(\Fun{x}{x \, N}\) when \(x \fresh N\). One may notice that it is the shape of a value translated to CPS. This is not essential for the right inverse theorem to hold but it is a necessary adjustment for the left post-inverse theorem.

When it comes to discovery of those translations, it was mostly trial and error, working within confines of a reflection, fine-tuning simultaneously the translations and the calculus. Just as CPS translations come in a pair: one on value, one on general terms, both had to be inverted and therefore we have in fact two different embeddings of lambda terms: one embeds into values of \(\Ld\) and the other into general terms of \(\Ld\).
They need each other to compute DS translation just like the two CPS translations interleave to bring a CPS term.

\begin{figure}[tp!]
  \label{fig:hash}
  \[ \# : \lambda \to \Ld \]
  \begin{minipage}{.5\textwidth}
    \begin{equation*}
      \begin{matrix*}[l]
        x^\# & \deq \Sh{x} & \\
        (\Fun{x}{x \, N})^\# & \deq N^\natural & \IF x \fresh N \\
        (\Fun{x}{M})^\# & \deq \Shift{x}{M^\#} & \OTH \\
        (M \, N)^\# & \deq M^\natural \, N^\natural & \\
      \end{matrix*}
    \end{equation*}
  \end{minipage}
  \begin{minipage}{.4\textwidth}
    \begin{equation*}
      \begin{matrix*}[l]
        x^\natural & \deq x & \\
        (\Fun{x}{M})^\natural & \deq \Fun{x}{M^\#} & \\
        (M N)^\natural & \deq \Dol{M^\natural \, N^\natural} & \\
        & \\
      \end{matrix*}
    \end{equation*}
  \end{minipage}
  \caption{Conversion from \(\lambda\) back to direct style \(\Ld\).}
\end{figure}

\begin{example}
  Let's take the CPS term from the previous example and reflect it into $\Ld$.
  \[(\Fun{k_1}{k_1 \, (\Fun{x\,k_2}{k_2 \, x})})^\# \equiv
    (\Fun{x\,k_2}{k_2 \, x})^\natural \equiv
    \Fun{x}{(\Fun{k_2}{k_2 \, x})^\#} \equiv
    \Fun{x}{x^\natural} \equiv
    \Fun{x}{x}.
  \]
\end{example}

\begin{theorem}[Right inverse of \(*\)]
  For all \(M \in \lambda\), \(M^{\#*} \equiv M\) and \(M^{\natural\dagger} \equiv M\).
  \label{thm:right-inverse}
\end{theorem}

\begin{proof}
  Structural induction on \(M\).
  \begin{itemize}
    \item Case \(M \equiv x\).
      \(x^{\#*} \equiv \Sh{x}^* \equiv x^\dagger \equiv x\) and
      \(x^{\natural\dagger} \equiv x^\dagger \equiv x\).
    \item Case \(M \equiv \Fun{x}{M_1}\). First identity has two subcases.
      \begin{itemize}
        \item Case \(M_1 \equiv x \, M_2\) and \(x \fresh M_2\).
          \((\Fun{x}{x \, M_2})^{\#*}
          \equiv M_2^{\natural*}
          \equiv \Fun{x}{x \, M_2^{\natural\dagger}}
          \equiv_\text{IH} \Fun{x}{x \, M_2}\).
          \item Opposite case.
          \((\Fun{x}{M_1})^{\#*}
          \equiv (\Shift{x}{M_1^\#})^*
          \equiv (\Fun{x}{M_1^\#})^\dagger
          \equiv \Fun{x}{M_1^{\#*}}
          \equiv_\text{IH} \Fun{x}{M_1}\).
      \end{itemize}
      Second identity:
      \((\Fun{x}{M_1})^{\natural\dagger}
      \equiv (\Fun{x}{M_1^\#})^\dagger
      \equiv \Fun{x}{M_1^{\#*}}
      \equiv_\text{IH} \Fun{x}{M_1}\).
    \item Case \(M \equiv M_1 \, M_2\).
      \((M_1 \, M_2)^{\#*}
      \equiv (M_1^\natural \, M_2^\natural)^*
      \equiv M_1^{\natural\dagger} \, M_2^{\natural\dagger}
      \equiv_\text{IH} M_1 \, M_2\),
      \((M_1 \, M_2)^{\natural\dagger}
      \equiv \Dol{M_1^\natural \, M_2^\natural}^\dagger
      \equiv (M_1^\natural \, M_2^\natural)^*
      \equiv M_1^{\natural\dagger} \, M_2^{\natural\dagger}
      \equiv_\text{IH} M_1 \, M_2\).
  \end{itemize}
\end{proof}

\subsection{Monotonicity of CPS and DS translation}

In this subsection we establish that both \(*\) and \(\#\) are monotone. A few lemmas are needed for these to work, notably substitution lemmas. Some lemmas could be proven by structural induction, but it is more concise and natural for them to follow recursion pattern of the definition of \(*\). For this reason, we prove these lemmas by induction on the size of \(M\). Since we have a grammar of valid terms, we simply define the size of a term as the number of nodes in a parse tree of that term. Crucially, variables are smaller in this sense than nonvalues, hence \(J[x]\) is always smaller than \(J[P]\).

\begin{lemma}[Values reduce to values]
  If \(V \to_\Ld N\), then for some \(W\), \(N \equiv W\).
\end{lemma}

\begin{proof}
  By definition of \(\to\), there exist \(M'\), \(N'\) and \(C\), such that,
    \(M' \mapsto N'\), \(V \equiv C[M']\) and \(N \equiv C[N']\).
  We proceed by cases on \(V\) and \(C\).
  \begin{itemize}
    \item Case \(V \equiv x\). Necessarily \(C \equiv [\,]\). There is no possible reduction, proposition holds vacuously.
    \item Case \(V \equiv \Fun{x}{M_1}\).
    \begin{itemize}
      \item Subcase \(C \equiv [\,]\). Necessarily \(V (\ev) N\). By definition of \(\ev\), \(V \equiv \Fun{x}{x \, W}\) and \(N \equiv W\) for some \(W\).
      \item Subcase \(C \equiv \Fun{x}{C'}\). \(N \equiv C[N'] \equiv \Fun{x}{C'[N']}\).
    \end{itemize}
    \item Case \(V \equiv \Dol{M_1}\).
    \begin{itemize}
      \item Subcase \(C \equiv [\,]\) and \(V (\dolv) N\). By definition of \(\dolv\), \(V \equiv \Dol{W}\) and \(N \equiv \Fun{x}{x \, W}\) for some \(W\).
      \item Subcase \(C \equiv [\,]\) and \(V (\dolsh) N\). By definition of \(\dolsh\), \(V \equiv \Dol{\Sh{W}}\) and \(N \equiv W\) for some \(W\).
      \item Subcase \(C \equiv \Dol{C'}\). \(N \equiv C[N'] \equiv \Dol{C'[N']}\).
    \end{itemize}
  \end{itemize}
\end{proof}

\begin{lemma}
  \label{lemma:asterisk-dagger}
  \begin{itemize}
    \item If \(V^* \to_\lambda W^*\), then \(V^\dagger \to_\lambda W^\dagger\).
  \end{itemize}
\end{lemma}

\begin{proof}
  Assume \(V^* \to W^*\).
  Unfolding the definition of \(*\) gives \(\Fun{k}{k \, V^\dagger} \to \Fun{k}{k \, W^\dagger}\).
  By definition of \(\to\) (in \(\lambda\)), there exist \(M\), \(N\) and \(C\), such that, \(M (\beta\cup\eta) N\), \(\Fun{k}{k \, V^\dagger} \equiv C[M]\) and \(\Fun{k}{k \, W^\dagger} \equiv C[N]\).
  We proceed by cases on \(C\).
  \begin{itemize}
    \item Case \(C \equiv [\,]\) or \(C \equiv \Fun{k}{[\,]}\). There is no possible reduction, proposition holds vacuously.
    \item Case \(C \equiv \Fun{k}{k \, C'}\). Given that \(M (\beta\cup\eta) N\), \(V^\dagger \equiv C'[M]\) and \(W^\dagger \equiv C'[N]\), we infer \(V^\dagger \to W^\dagger\).
  \end{itemize}
\end{proof}

\begin{lemma}[Substitution lemma for \(*\) and \(\dagger\)]
  \label{lemma:subst-asterisk}
  For every \(M\) and \(V\),
  \begin{itemize}
    \item \(M^*[V^\dagger/x] \equiv M[V/x]^*\) and
    \item if \(M\) is a value, then \(M^\dagger[V^\dagger/x] \equiv M[V/x]^\dagger\).
  \end{itemize}
\end{lemma}

\begin{proof}
  By induction on the size of \(M\).
  We prove the second conjunct below.
  \begin{itemize}
    \item Case \(M \equiv x\).
    \(
      x^\dagger[V^\dagger/x]
      \equiv x[V^\dagger/x]
      \equiv V^\dagger
      \equiv x[V/x]^\dagger
    \)
    \item Case \(M \equiv y\).
    \(
      y^\dagger[V^\dagger/x]
      \equiv y[V^\dagger/x]
      \equiv y
      \equiv y^\dagger
      \equiv y[V/x]^\dagger
    \).
    \item Case \(M \equiv \Fun{y}{M_1}\).
    \(
      (\Fun{y}{M_1})^\dagger[V^\dagger/x]
      \equiv \Fun{y}{M_1^*[V^\dagger/x]}
      \equiv_\text{IH} \Fun{y}{M_1[V/x]^*}
      \equiv (\Fun{y}{M_1})[V/x]^\dagger
    \).
    \item Case \(M \equiv \Dol{M_1}\).
    \(
      \Dol{M_1}^\dagger[V^\dagger/x]
      \equiv M_1^*[V^\dagger/x]
      \equiv_\text{IH} M_1[V/x]^*
      \equiv \Dol{M_1}[V/x]^\dagger
    \).
  \end{itemize}
  We now prove the first conjunct.
  \begin{itemize}
    \item Case \(M \equiv J[P]\).
    \[
      J[P]^*[V^\dagger/x]
      \equiv \Letk{y}{P^*[V^\dagger/x]}{J[y]^*[V^\dagger/x]}
      \equiv_\text{IH} \Letk{y}{P[V/x]^*}{J[y][V/x]^*}
      \equiv J[P][V/x]^*
    \]
    \item Case \(M \equiv V_1 \, V_2\).
    \(
      (V_1 \, V_2)^*[V^\dagger/x]
      \equiv V_1^\dagger[V^\dagger/x] \, V_2^\dagger[V^\dagger/x]
      \equiv_\text{IH} V_1[V/x]^\dagger \, V_2[V/x]^\dagger
      \equiv (V_1 \, V_2)[V/x]^*
    \)
    \item Case \(M \equiv \Sh{V_1}\).
    \(
      \Sh{V_1}^*[V^\dagger/x]
      \equiv V_1^\dagger[V^\dagger/x]
      \equiv_\text{IH} V_1[V/x]^\dagger
      \equiv \Sh{V_1}[V/x]^*
    \)
    \item Case \(M \equiv W\). For brevity, we consider all values at once and reuse the proofs of the second conjunct.
    \(
      W^*[V^\dagger/x]
      \equiv \Fun{k}{k \, W^\dagger[V^\dagger/x]}
      \equiv_\text{the second conjunct} \Fun{k}{k \, W[V/x]^\dagger}
      \equiv W[V/x]^*
    \).
  \end{itemize}
\end{proof}

\begin{lemma}
  \label{lemma:contraction-asterisk}
  If \(M_1 \mapsto_\Ld M_2\), then \(M_1^* \to_\lambda^? M_2^*\).
\end{lemma}

\begin{proof}
  \begin{itemize}
    \item Case \(\bv\).
    \(
      ((\Fun{x}{M}) \, V)^*
      \equiv (\Fun{x}{M^*}) \, V^\dagger
      \to M^*[V^\dagger/x]
      \equiv_\text{Lemma~\ref{lemma:subst-asterisk}} M[V/x]^*
    \).
    \item Case \(\ev\).
    \(
      (\Fun{x}{V \, x})^*
      \equiv \Fun{k}{k \, (\Fun{x}{V^\dagger \, x})}
      \to_\eta \Fun{k}{k \, V^\dagger}
      \equiv V^*
    \).
    \item Case \(\dolv\).
    \(
      \Dol{V}^*
      \equiv \Fun{k}{k \, \Dol{V}^\dagger}
      \equiv \Fun{k}{k \, V^*}
      \equiv \Fun{k}{k \, (\Fun{x}{x \, V^\dagger})}
      \equiv \Fun{k}{k \, (\Fun{x}{x \, V})^\dagger}
      \equiv (\Fun{x}{x \, V})^*
    \).
    \item Case \(\dolsh\).
    \(
      \Dol{\Sh{V}}^*
      \equiv \Fun{k}{k \, \Dol{\Sh{V}}^\dagger}
      \equiv \Fun{k}{k \, \Sh{V}^*}
      \equiv \Fun{k}{k \, V^\dagger}
      \equiv V^*
    \).
    \item Case \(\shdol\).
    \(
      \Sh{\Dol{M}}^*
      \equiv \Dol{M}^\dagger
      \equiv M^*
    \).
    \item Case \(\pure\).
    \(
      (\Shift{x}{x \, V})^*
      \equiv \Fun{x}{(x \, V)^*}
      \equiv \Fun{x}{x \, V^\dagger}
      \equiv V^*
    \).
    \item Case \(\bind\).
    \begin{align*}
      J[P]^*
      & \equiv \Fun{k}{P^* \, (\Fun{x}{J[x]^* \, k})}
      \equiv \Fun{k}{\Dol{P}^\dagger \, (\Fun{x}{\Dol{J[x]}^\dagger \, k})}
      \equiv \Fun{k}{\Dol{P}^\dagger \, (\Fun{x}{k \DOL J[x]})^\dagger} \\
      & \equiv (\Fun{k}{(\Fun{x}{k \DOL J[x]}) \DOL P})^\dagger
      \equiv (\Shift{k}{(\Fun{x}{k \DOL J[x]}) \DOL P})^*
      \equiv (\Let{x}{P}{J[x]})^*
    \end{align*}
  \end{itemize}
\end{proof}

\begin{lemma}[Single-step reduction preservation by \(*\)]
  \label{lemma:single-asterisk}
  If \(M \to_\Ld N\), then \(M^* \to_\lambda^? N^*\).
\end{lemma}

\begin{proof}
  By structural induction on \(M\).
  We unfold the definition of \(\to\) (in \(\Ld\)):
  there exist \(M'\), \(N'\) and \(C\), such that,
    \(M' \mapsto N'\), \(M \equiv C[M']\) and \(N \equiv C[N']\).
  If \(M\) is a variable, then there is no possible reduction and the proposition holds vacuously.
  If \(C \equiv [\,]\), then apply Lemma~\ref{lemma:contraction-asterisk}.
  Otherwise, we proceed by cases on \(M\) and (nonempty) \(C\).
  In each case below we consider some proper subcontext \(C'\) of \(C\). Because \(C'[M']\) is a proper subterm of \(M\), then from inductive hypothesis and \(C'[M'] \to C'[N']\) we infer \(C'[M']^* \to^? C'[N']^*\). When \(C'[M']\) is a value, it reduces to a value \(C'[N']\), so \(C'[N']^\dagger\) is well-defined and Lemma~\ref{lemma:asterisk-dagger} gives us \(C'[M']^\dagger \to^? C'[N']^\dagger\).
  \begin{itemize}
    \item Case \(M \equiv \Fun{x}{M_1}\) and \(C \equiv \Fun{x}{C'}\).
    \(
      (\Fun{x}{M_1})^*
      \equiv \Fun{k}{k \, (\Fun{x}{M_1^*})}
      \equiv \Fun{k}{k \, (\Fun{x}{C'[M']^*})}
      \to^?_\text{IH} \Fun{k}{k \, (\Fun{x}{C'[N']^*})}
      \equiv (\Fun{x}{C'[N']})^*
      \equiv N^*
    \).
    \item Case \(M \equiv \Dol{M_1}\) and \(C \equiv \Dol{C'}\).
    \(
      \Dol{M_1}^*
      \equiv \Fun{k}{k \, M_1^*)}
      \equiv \Fun{k}{k \, C'[M']^*}
      \to^?_\text{IH} \Fun{k}{k \, C'[N']^*}
      \equiv (\Dol{C'[N']})^*
      \equiv N^*
    \).
    \item Case \(M \equiv \Sh{V}\) and \(C \equiv \Sh{C'}\).
    \(
      \Sh{V}^*
      \equiv V^\dagger
      \equiv C'[M']^\dagger
      \to^?_\text{IH} C'[N']^\dagger
      \equiv \Sh{C'[N']}^*
      \equiv N^*
    \).
    \item Case \(M \equiv J[P]\) and \(C \equiv J[C']\).
    \(
      J[P]^*
      \equiv \Fun{k}{P^* \, (\Fun{x}{J[x]^* \, k})}
      \equiv \Fun{k}{C'[M']^* \, (\Fun{x}{J[x]^* \, k})}
      \to^?_\text{IH} \Fun{k}{C'[N']^* \, (\Fun{x}{J[x]^* \, k})}
      \equiv J[C'[N']]^*
      \equiv N^*
    \).
    \item Case \(M \equiv V_1 \, V_2\) and \(C \equiv C' \, V_2\).
    \(
      (V_1 \, V_2)^*
      \equiv V_1^\dagger \, V_2^\dagger
      \equiv C'[M']^\dagger \, V_2^\dagger
      \to^?_\text{IH} C'[N']^\dagger \, V_2^\dagger
      \equiv (C'[N'] \, V_2)^*
      \equiv N^*
    \).
    \item Case \(M \equiv V_1 \, V_2\) and \(C \equiv V_1 \, C'\).
    \(
      (V_1 \, V_2)^*
      \equiv V_1^\dagger \, V_2^\dagger
      \equiv V_1^\dagger \, C'[M']^\dagger
      \to^?_\text{IH} V_1^\dagger \, C'[N']^\dagger
      \equiv (V_1 \, C'[N'])^*
      \equiv N^*
    \).
    \item Case \(M \equiv V_1 \, P_2\) and \(C \equiv C' \, P_2\).
    \(
      (V_1 \, P_2)^*
      \equiv \Letk{x}{P_2^*}{V_1^\dagger \, x}
      \equiv \Letk{x}{P_2^*}{C'[M']^\dagger \, x}
      \to^?_\text{IH} \Letk{x}{P_2^*}{C'[N']^\dagger \, x}
      \equiv (C'[N'] \, P_2)^*
      \equiv N^*
    \).
    \item Case \(M \equiv P_1 \, V_2\) and \(C \equiv P_1 \, C'\).
    \(
      (P_1 \, V_2)^*
      \equiv \Letk{x}{P_1^*}{x \, V_2^\dagger}
      \equiv \Letk{x}{P_1^*}{x \, C'[M']^\dagger}
      \to^?_\text{IH} \Letk{x}{P_1^*}{x \, C'[N']^\dagger}
      \equiv (P_1 \, C'[N'])^*
      \equiv N^*
    \).
    \item Case \(M \equiv P_1 \, P_2\) and \(C \equiv P_1 \, C'\).
    \begin{align*}
      (P_1 \, P_2)^*
      & \equiv \Letk{x}{P_1^*}{\Letk{y}{P_2^*}{x \, y}}
      \equiv \Letk{x}{P_1^*}{\Letk{y}{C'[M']^*}{x \, y}} \\
      & \to^?_\text{IH} \Letk{x}{P_1^*}{\Letk{y}{C'[N']^*}{x \, y}}
      \equiv (P_1 \, C'[N'])^*
      \equiv N^*
    \end{align*}
  \end{itemize}
\end{proof}

\begin{theorem}[Monotonicity of \(*\)]
  \label{thm:mono-asterisk}
  For all \( M, N \in \Ld \),
  \( M \redu_\Ld N \) implies
  \( M^* \redu_\lambda N^* \).
\end{theorem}

\begin{proof}
  By induction on the number \(n\) of reduction steps in \(M \redu N\).
  \begin{itemize}
    \item Case \(n = 0\), \(M \equiv N\). \(M^* \redu N^*\) in \(0\) steps.
    \item Case \(n > 0\), \(M \to^{n-1} L \to N\) for some \(L\).
    \(
      M^*
      \redu_\text{IH} L^*
      \to^?_\text{Lemma~\ref{lemma:single-asterisk}} N^*
    \).
  \end{itemize}
\end{proof}

\begin{lemma}
  \label{lemma:sh-nat}
  For all \(M \in \lambda\), \(\Sh{M^\natural} \to_\Ld^? M^\#\).
\end{lemma}

\begin{proof}
  By cases on \(M\).
  \begin{itemize}
    \item Case \(M \equiv x\).
    \(
      \Sh{x^\natural}
      \equiv \Sh{x} 
      \equiv x^\#
    \). 
    \item Case \(M \equiv \Fun{x}{M_1}\).
    \begin{itemize}
      \item Case \(M_1 \equiv x \, M_2\) and \(x \fresh M_2\).
      \(
        \Sh{(\Fun{x}{x \, M_2})^\natural}
        \equiv \Shift{x}{x \, M_2^\natural}
        \to_\pure M_2^\natural
        \equiv (\Fun{x}{x \, M_2})^\#
      \).
      \item Opposite case.
      \(
        \Sh{(\Fun{x}{M_1})^\natural}
        \equiv \Shift{x}{M_1^\#}
        \equiv (\Fun{x}{M_1})^\#
      \).
    \end{itemize}
    \item Case \(M \equiv M_1 \, M_2\).
    \(
      \Sh{(M_1 \, M_2)^\natural}
      \equiv \Sh{\Dol{M_1^\natural \, M_2^\natural}}
      \to_\shdol M_1^\natural \, M_2^\natural
      \equiv (M_1 \, M_2)^\#
    \).
  \end{itemize}
\end{proof}

\begin{lemma}
  \label{lemma:dol-hash}
  For all \(M \in \lambda\), \(\Dol{M^\#} \to_\Ld^? M^\natural\).
\end{lemma}

\begin{proof}
  By cases on \(M\).
  \begin{itemize}
    \item Case \(M \equiv x\).
    \(
      \Dol{x^\#}
      \equiv \Dol{\Sh{x}} 
      \to_\dolsh x
      \equiv x^\natural
    \). 
    \item Case \(M \equiv \Fun{x}{M_1}\).
    \begin{itemize}
      \item Case \(M_1 \equiv x \, M_2\) and \(x \fresh M_2\).
      \(
        \Dol{(\Fun{x}{x \, M_2})^\#}
        \equiv \Dol{M_2^\natural}
        \to_\dolv \Fun{x}{x \, M_2^\natural}
        \equiv (\Fun{x}{x \, M_2})^\natural
      \).
      \item Opposite case.
      \(
        \Dol{(\Fun{x}{M_1})^\#}
        \equiv \Dol{\Shift{x}{M_1^\#}}
        \to_\dolsh \Fun{x}{M_1^\#}
        \equiv (\Fun{x}{M_1})^\natural
      \).
    \end{itemize}
    \item Case \(M \equiv M_1 \, M_2\).
    \(
      \Dol{(M_1 \, M_2)^\#}
      \equiv \Dol{M_1^\natural \, M_2^\natural}
      \equiv (M_1 \, M_2)^\natural
    \).
  \end{itemize}
\end{proof}

\begin{lemma}
  \label{lemma:subst-hash}
  \(M^\#[N^\natural/x] \redu_\Ld M[N/x]^\#\) and
  \(M^\natural[N^\natural/x] \redu_\Ld M[N/x]^\natural\) hold.
\end{lemma}

\begin{proof}
  In the appendix.
\end{proof}

\begin{lemma}
  \label{lemma:contraction-hash}
  If \(M \beta N\) or \(M \eta N\), then
  \begin{itemize}
    \item \(M^\# \redu_\Ld N^\#\) and
    \item \(M^\natural \redu_\Ld N^\natural\).
  \end{itemize}
\end{lemma}

\begin{proof}
  \begin{itemize}
    \item Case \(\beta\).
    \begin{itemize}
      \item \(
        ((\Fun{x}{M}) \, N)^\#
        \equiv (\Fun{x}{M^\#}) \, N^\natural
        \to_\bv M^\#[N^\natural/x]
        \to_\text{Lemma~\ref{lemma:subst-hash}} M[N/x]^\#
      \).
      \item \(
        ((\Fun{x}{M}) \, N)^\natural
        \equiv \Dol{(\Fun{x}{M^\#}) \, N^\natural}
        \to_\bv \Dol{M^\#[N^\natural/x]}
        \to_\text{Lemma~\ref{lemma:subst-hash}} \Dol{M[N/x]^\#}
        \to^?_\text{Lemma~\ref{lemma:dol-hash}} M[N/x]^\natural
      \).
    \end{itemize}
    \item Case \(\eta\).
    \begin{itemize}
      \item \(
        (\Fun{x}{M \, x})^\#
        \equiv \Sh{\Fun{x}{M^\natural \, x}}
        \to_\ev \Sh{M^\natural}
        \to^?_\text{Lemma~\ref{lemma:sh-nat}} M^\#
      \).
      \item \(
        (\Fun{x}{M \, x})^\natural
        \equiv \Fun{x}{M^\natural \, x}
        \to_\ev M^\natural
      \).
    \end{itemize}
  \end{itemize}
\end{proof}

\begin{lemma}[Single-step reduction preservation
  by \(\#\) and \(\natural\)]
  \label{lemma:single-hash}
  If \(M \to_\lambda N\), then
  \begin{itemize}
    \item \(M^\# \redu_\Ld N^\#\) and
    \item \(M^\natural \redu_\Ld N^\natural\).
  \end{itemize}
  \label{thm:ssred-Ldt}
\end{lemma}

\begin{proof}
  By structural induction on \(M\).
  We unfold the definition of \(\to\) (in \(\lambda\)):
  there exist \(M'\), \(N'\) and \(C\), such that,
    \(M' \mapsto N'\), \(M \equiv C[M']\) and \(N \equiv C[N']\).
  If \(M\) is a variable, then there is no possible reduction and the proposition holds vacuously.
  If \(C \equiv [\,]\), then apply Lemma~\ref{lemma:contraction-hash}.
  Otherwise, we proceed by cases on \(M\) and (nonempty) \(C\).
  In each case below we consider some proper subcontext \(C'\) of \(C\). Because \(C'[M']\) is a proper subterm of \(M\), then from inductive hypothesis and \(C'[M'] \to C'[N']\) we infer \(C'[M']^\# \redu C'[N']^\#\) and \(C'[M']^\natural \redu C'[N']^\natural\).
  \begin{itemize}
    \item Case \(M \equiv M_1 \, M_2\) and \(C \equiv C' \, M_2\).
    \begin{itemize}
      \item \(
        (M_1 \, M_2)^\#
        \equiv M_1^\natural \, M_2^\natural
        \equiv C'[M']^\natural \, M_2^\natural
        \redu_\text{IH} C'[N']^\natural \, M_2^\natural
        \equiv (C'[N'] \, M_2)^\#
        \equiv N^\#
      \).
      \item \(
        (M_1 \, M_2)^\natural
        \equiv \Dol{M_1^\natural \, M_2^\natural}
        \equiv \Dol{C'[M']^\natural \, M_2^\natural}
        \redu_\text{IH} \Dol{C'[N']^\natural \, M_2^\natural}
        \equiv (C'[N'] \, M_2)^\natural
        \equiv N^\natural
      \). 
    \end{itemize}
    \item Case \(M \equiv M_1 \, M_2\) and \(C \equiv M_1 \, C'\).
    \begin{itemize}
      \item \(
        (M_1 \, M_2)^\#
        \equiv M_1^\natural \, M_2^\natural
        \equiv M_1^\natural \, C'[M']^\natural
        \redu_\text{IH} M_1^\natural \, C'[N']^\natural
        \equiv (M_1 \, C'[N'])^\#
        \equiv N^\#
      \).
      \item \(
        (M_1 \, M_2)^\natural
        \equiv \Dol{M_1^\natural \, M_2^\natural}
        \equiv \Dol{M_1^\natural \, C'[M']^\natural}
        \redu_\text{IH} \Dol{M_1^\natural \, C'[N']^\natural}
        \equiv (M_1 \, C'[N'])^\natural
        \equiv N^\natural
      \).
    \end{itemize}
    \item Case \(M \equiv \Fun{x}{M_1}\) and \(C \equiv \Fun{x}{C'}\), second conjunct.
    \(
      (\Fun{x}{M_1})^\natural
      \equiv \Fun{x}{M_1^\#}
      \equiv \Fun{x}{C'[M']^\#}
      \redu_\text{IH} \Fun{x}{C'[N']^\#}
      \equiv (\Fun{x}{C'[N']})^\natural
      \equiv N^\natural
    \).
    \item Case \(M \equiv \Fun{x}{M_1}\), first conjunct.
    \begin{itemize}
      \item Case \(M_1 \equiv x \, M_2\), \(x \fresh M_2\) and \(C \equiv \Fun{x}{x \, C'}\).
      \(
        (\Fun{x}{x \, M_2})^\#
        \equiv M_2^\natural
        \equiv C'[M']^\natural
        \redu_\text{IH} C'[N']^\natural
        \equiv (\Fun{x}{x \, C'[N']})^\#
        \equiv N^\#
      \).
      \item Opposite case, \(C \equiv \Fun{x}{C'}\).
      \begin{align*}
        (\Fun{x}{M_1})^\#
        & \equiv \Shift{x}{M_1^\#}
        \equiv \Shift{x}{C'[M']^\#}
        \redu_\text{IH} \Shift{x}{C'[N']^\#} \\
        & \equiv \Sh{(\Fun{x}{C'[N']})^\natural}
        \to^?_\text{Lemma~\ref{lemma:sh-nat}} (\Fun{x}{C'[N']})^\#
      \end{align*}
    \end{itemize}
  \end{itemize}
\end{proof}

\begin{theorem}[Monotonicity of \(\#\)]
  \label{thm:mono-hash}
  For all \(M, N \in \lambda\),
  if \(M \redu N\), then \(M^\# \redu N^\#\).
\end{theorem}

\begin{proof}
  By induction on the number \(n\) of reduction steps in \(M \redu N\).
  \begin{itemize}
    \item Case \(n = 0\), \(M \equiv N\). \(M^* \redu N^*\) in \(0\) steps.
    \item Case \(n > 0\), \(M \to^{n-1} L \to N\) for some \(L\).
    \(
      M^*
      \redu_\text{IH} L^*
      \redu_\text{Lemma~\ref{lemma:single-hash}} N^*
    \).
  \end{itemize}
\end{proof}

\subsection{The reflection theorem and the kernel calculus}

The last missing ingredient of reflection is the left post-inverse theorem. Inspection of the proof reveals that neither \(\bv\) nor \(\ev\) is used while every other rule is used in some instance.

\begin{theorem}[Left post-inverse of \( * \)]
  \label{thm:left-inverse}
  For all \( M \in \Ld \),
  \begin{itemize}
    \item \(M \redu M^{*\#}\) and
    \item if \(M\) is a value, then \(M \redu M^{\dagger\natural}\).
  \end{itemize}
\end{theorem}

\begin{proof}
  By induction on the size of \(M\).
  We prove the second conjunct below.
  \begin{itemize}
    \item Case \(M \equiv x\).
    \(
      x
      \equiv x^\natural
      \equiv x^{\dagger\natural}
    \).
    \item Case \(M \equiv \Fun{x}{M_1}\).
    \(
      \Fun{x}{M_1}
      \redu_\text{IH} \Fun{x}{M_1^{*\#}}
      \equiv (\Fun{x}{M_1^*})^\natural
      \equiv (\Fun{x}{M_1})^{\dagger\natural}
    \).
    \item Case \(M \equiv \Dol{M_1}\).
    \(
      \Dol{M_1}
      \redu_\text{IH} \Dol{M_1^{*\#}}
      \to^?_\text{Lemma~\ref{lemma:dol-hash}} M_1^{*\natural}
      \equiv \Dol{M_1}^{\dagger\natural}
    \).
  \end{itemize}
  We now prove the first conjunct.
  \begin{itemize}
    \item Case \(M \equiv J[P]\).
    \begin{align*}
      J[P]
      & \to_\bind \Let{x}{P}{J[x]}
      \redu_\text{IH} \Let{x}{P^{*\#}}{J[x]^{*\#}}
      \equiv \Shift{k}{\Dol{P^{*\#}} \, (\Fun{x}{\Dol{J[x]^{*\#}} \, k})} \\
      & \to^?_\text{Lemma~\ref{lemma:dol-hash}} \Shift{k}{P^{*\natural} \, (\Fun{x}{J[x]^{*\natural} \, k})}
      \equiv (\Letk{x}{P^*}{J[x]^*})^\#
      \equiv J[P]^{*\#}
    \end{align*}
    \item Case \(M \equiv V_1 \, V_2\).
    \(
      V_1 \, V_2
      \redu_\text{IH} V_1^{\dagger\natural} \, V_2^{\dagger\natural}
      \equiv (V_1^\dagger \, V_2^\dagger)^\#
      \equiv (V_1 \, V_2)^{*\#}
    \).
    \item Case \(M \equiv \Sh{V_1}\).
    \(
      \Sh{V_1}
      \redu_\text{IH} \Sh{V_1^{\dagger\natural}}
      \to^?_\text{Lemma~\ref{lemma:sh-nat}} V_1^{\dagger\#}
      \equiv \Sh{V_1}^{*\#}
    \).
    \item Case \(M \equiv W\). For brevity, we consider all values at once and reuse the proofs of the second conjunct.
    \(
      W
      \redu_\text{the second conjunct} W^{\dagger\natural}
      \equiv (\Fun{k}{k \, W^\dagger})^\#
      \equiv W^{*\#}
    \).
  \end{itemize}
\end{proof}

\begin{example}
  The CPS translation of $S$-combinator $\Fun{x\,y\,z}{x\,z\,(y\,z)}$ is \[S^* \equiv \Fun{k_1}{k_1\,(\Fun{x\,k_2}{k_2\,(\Fun{y\,k_3}{k_3\,(\Fun{z\,k_4}{x\,y\,(\Fun{f}{(\Fun{k_5}{y\,z\,(\Fun{a}{f\,a\,k_5})})\,k_4})})})})}.\]
  Reflection of this CPS term back to $\Ld$ is
  \(S^{*\#} \equiv \Fun{x\,y\,z}{\Shift{k_4}{(\Fun{f}{(\Fun{k_5}{(\Fun{a}{k_5 \DOL f\,a}) \DOL y\,z})\,k_4}) \DOL x\,y}}\).
  One can see that $S \redu_\bind \Fun{x\,y\,z}{\Let{f}{x\,y}{\Let{a}{y\,z}{f\,a}}} \to_\dolsh S^{*\#}$, in accordance with Theorem~\ref{thm:left-inverse}.
\end{example}

\begin{theorem}[\((*, \#)\) form a Galois connection]
  For any \(M \in \Ld\) and \(N \in \lambda\) we have \(M \redu_\Ld N^\#\)
  if and only if \(M^* \redu_\lambda N\).
\end{theorem}

\begin{proof}
  The left-to-right implication follows from
  the monotonicity of \(*\) (Theorem~\ref{thm:mono-asterisk})
  and the right inverse property (Theorem~\ref{thm:right-inverse}).
  The right-to-left implication follows from
  the monotonicity of \(\#\) (Theorem~\ref{thm:mono-hash}).
  and the left post-inverse property (Theorem~\ref{thm:left-inverse}).
\end{proof}

\begin{corollary}[Reflection]
  Since \(\#\) is a right inverse of \(*\), the Galois connection is
  indeed a reflection.
\end{corollary}

The range of \(\#\) translation (from now on, the \emph{kernel}, or \(\lambda^\#\)) has an induced preorder \[\redu_{\lambda^\#} = \left\{ (M, N) \in \lambda^\# \times \lambda^\# \middle| M \redu_\Ld N \right\}.\] The kernel can be considered a calculus in its own right and it turns out we know this calculus very well: it is preorder-isomorphic to \(\lambda\) via \(*\) restricted to kernel and its inverse \(\#\). What that means is that the kernel is a fragment of \(\Ld\) which behaves exactly like \(\lambda\)-calculus (\(M \redu_{\lambda^\#} N\) iff \(M^* \redu_\lambda N^*\)) except that terms look different and \(*\) allows to decode these to familiar \(\lambda\)-terms. We believe that reductions in kernel calculus could be fully characterized, in similar fashion to ~\cite{DBLP:journals/toplas/SabryW97,DBLP:conf/fscd/BiernackiPS20,DBLP:conf/ppdp/BiernackiPS21}, but we did not investigate that yet. Experience suggest that these reductions would be quite verbose.

\begin{corollary}
  Calculi \(\Lmv\)~\cite{DBLP:conf/icfp/DownenA14}, \(\ld\)~\cite{DBLP:conf/csl/Materzok13}, \(\Ld\) and \(\lambda\) are in equational correspondence.
\end{corollary}

\begin{proof}
  Calculus \(\Lmv\) is a calculus created by relaxing some syntactic restrictions of \(\lmv\). \(\Lmv\) was defined and proven to correspond equationally with \(\ld\) in~\cite{DBLP:conf/icfp/DownenA14}.
  Calculus \(\ld\) was defined by Materzok in~\cite{DBLP:conf/csl/Materzok13}.
  We proved its correspondence with \(\Ld\) in Section~\ref{sec:well-defined}.
  In this section we proved that \(\lambda\) can be reflected into \(\Ld\).
  This fact can be weakened to equational correspondence.
  All of \(\Lmv\), \(\ld\), \(\Ld\) and \(\lambda\) correspond equationally because equational correspondences are composable.
\end{proof}

\section{Conclusion}
\label{sec:conclusion}

We developed a novel calculus of delimited control \(\Ld\) whose confluent reduction theory induces an axiomatisation that is both sound and complete with respect to denotational semantics provided by CPS translation.
We have shown a close relationship (i.e. reflection) between reduction theories of pure lambda calculus and calculus \(\Ld\), which justifies the assertion that \texttt{shift0}-style delimited control and call-by-name are diffent means of expressing the same computation.

Sabry and Felleisen~\cite{DBLP:journals/lisp/SabryF93} found an undirected axiomatisation of pure call-by-value lambda calculus and also of \texttt{call/cc}-style abortive control. Due to lack of interest in directedness, they amalgamated some rules of opposite direction into \(\beta_\Omega\) axiom. Kameyama and Hasegawa~\cite{DBLP:conf/icfp/KameyamaH03} found an undirected axiomatisation of \texttt{shift}-style delimited control, while Materzok~\cite{DBLP:conf/csl/Materzok13} axiomatised \texttt{shift0}. Our paper spells out the directed axiomatisation for the latter. Our rules are not only directed  but also local (or fine-grained) which means that no recursive family of contexts is needed to express them.

Developments in this paper are untyped. We conjecture that it should be possible to adapt existing type systems for \texttt{shift0}~\cite{DBLP:conf/icfp/MaterzokB11,DBLP:conf/csl/Materzok13,DBLP:conf/rta/PirogPS19}. We are advancing in the process of axiomatising algebraic effects with deep handlers, which are a related~\cite{DBLP:journals/jfp/0002KLP19,DBLP:conf/rta/PirogPS19} approach to computational effects.

\bibliographystyle{entics}
\bibliography{references}

\appendix

\section{Extra proofs}

\begin{lemma}[\ref{lemma:app-asterisk}]
  Equality \((M \, N)^* = \Fun{k}{M^* \, (\Fun{x}{N^* \, (\Fun{y}{x \, y \, k})})}\) holds in \(\lambda\).
\end{lemma}

\begin{proof}
  By cases (value/nonvalue) on \(M\) and \(N\).
  \begin{itemize}
    \item Case \(M \equiv V\), \(N \equiv W\).
    \begin{align*}
      (V \, W)^*
      & \equiv V^\dagger \, W^\dagger
      =_\eta \Fun{k}{V^\dagger \, W^\dagger \, k}
      =_\beta \Fun{k}{(\Fun{x}{x \, W^\dagger \, k}) \, V^\dagger}
      =_\beta \Fun{k}{V^* \, (\Fun{x}{x \, W^\dagger \, k})} \\
      & =_\beta \Fun{k}{V^* \, (\Fun{x}{(\Fun{y}{x \, y \, k}) \, W^\dagger})}
      =_\beta \Fun{k}{V^* \, (\Fun{x}{W^* \, (\Fun{y}{x \, y \, k})})}
    \end{align*}
    \item Case \(M \equiv V\), \(N \equiv Q\).
    \begin{equation*}
      (V \, Q)^*
      \equiv \Letk{y}{Q^*}{V^\dagger \, y}
      =_\beta \Fun{k}{(\Fun{x}{Q^* \, (\Fun{y}{x \, y \, k}})) \, V^\dagger}
      =_\beta \Fun{k}{V^* \, (\Fun{x}{Q^* \, (\Fun{y}{x \, y \, k})})}
    \end{equation*}
    \item Case \(M \equiv P\), \(N \equiv W\).
    \begin{equation*}
      (P \, W)^*
      \equiv \Letk{x}{P^*}{x \, W^\dagger}
      =_\beta \Fun{k}{P^* \, (\Fun{x}{(\Fun{y}{x \, y \, k}) \, W^\dagger})}
      =_\beta \Fun{k}{P^* \, (\Fun{x}{W^* \, (\Fun{y}{x \, y \, k})})}
    \end{equation*}
    \item Case \(M \equiv P\), \(N \equiv Q\).
    \begin{equation*}
      (P \, Q)^*
      \equiv \Letk{x}{P^*}{(x \, Q)^*}
      \equiv \Letk{x}{P^*}{(\Letk{y}{Q^*}{x \, y})}
      =_\beta \Fun{k}{P^* \, (\Fun{x}{Q^* \, (\Fun{y}{x \, y \, k})})}
    \end{equation*}
  \end{itemize}
\end{proof}

\begin{lemma}[\ref{lemma:subst-hash}]
  \(M^\#[N^\natural/x] \redu_\Ld M[N/x]^\#\) and
  \(M^\natural[N^\natural/x] \redu_\Ld M[N/x]^\natural\) hold.
\end{lemma}

\begin{proof}
  By structural induction on \(M\).
  We prove the first conjunct below.
  \begin{itemize}
    \item Case \(M \equiv x\).
    \(
      x^\#[N^\natural/x]
      \equiv \Sh{x}[N^\natural/x]
      \equiv \Sh{N^\natural}
      \to^?_\text{Lemma~\ref{lemma:sh-nat}} N^\#
      \equiv x[N/x]^\#
    \).
    \item Case \(M \equiv y\).
    \(
      y^\#[N^\natural/x]
      \equiv \Sh{y}[N^\natural/x]
      \equiv \Sh{y}
      \equiv y^\#
      \equiv y[N/x]^\#
    \).
    \item Case \(M \equiv M_1 \, M_2\).
    \[
      (M_1 \, M_2)^\#[N^\natural/x]
      \equiv M_1^\natural[N^\natural/x] \, M_2^\natural[N^\natural/x]
      \redu_\text{IH} M_1[N/x]^\natural \, M_2[N/x]^\natural
      \equiv (M_1 \, M_2)[N/x]^\#
    \]
    \item Case \(M \equiv \Fun{y}{M_1}\).
    If \(M_1 \equiv y \, M_2\) and \(y \fresh M_2\), then
    \[
      (\Fun{y}{y \, M_2})^\#[N^\natural/x]
      \equiv M_2^\natural[N^\natural/x]
      \redu_\text{IH} M_2[N/x]^\natural
      \equiv (\Fun{y}{y \, M_2})[N/x]^\#.
    \]
    Otherwise,
    \(
      (\Fun{y}{M_1})^\#[N^\natural/x]
      \equiv \Shift{y}{M_1^\#[N^\natural/x]}
      \redu_\text{IH} \Shift{y}{M_1[N/x]^\natural}
      \to^?_\text{Lemma~\ref{lemma:sh-nat}} (\Fun{y}{M_1[N/x]})^\#
      \equiv (\Fun{y}{M_1})[N/x]^\#
    \).
  \end{itemize}
  We now prove the second conjunct.
  \begin{itemize}
    \item Case \(M \equiv x\).
    \(
      x^\natural[N^\natural/x]
      \equiv x[N^\natural/x]
      \equiv N^\natural
      \equiv x[N/x]^\natural
    \).
    \item Case \(M \equiv y\).
    \(
      y^\natural[N^\natural/x]
      \equiv y[N^\natural/x]
      \equiv y
      \equiv y^\natural
      \equiv y[N/x]^\natural
    \).
    \item Case \(M \equiv M_1 \, M_2\).
    \[
      (M_1 \, M_2)^\natural[N^\natural/x]
      \equiv \Dol{M_1^\natural[N^\natural/x] \, M_2^\natural[N^\natural/x]}
      \redu_\text{IH} \Dol{M_1[N/x]^\natural \, M_2[N/x]^\natural}
      \equiv (M_1 \, M_2)[N/x]^\natural
    \]
    \item Case \(M \equiv \Fun{y}{M_1}\).
    \(
      (\Fun{y}{M_1})^\natural[N^\natural/x]
      \equiv \Fun{y}{M_1^\#[N^\natural/x]}
      \redu_\text{IH} \Fun{y}{M_1[N/x]^\#}
      \equiv (\Fun{y}{M_1})[N/x]^\natural
    \).
  \end{itemize}
\end{proof}

\end{document}